\pdfoutput=1
\documentclass{article}
\usepackage{amsmath,amssymb,graphicx,enumerate,bbm,color,a4wide}
\newcommand{\homepage}[1]{{\textit{Web:} \texttt{#1}}}
\newcommand{\tmem}[1]{{\em #1\/}}
\newcommand{\tmop}[1]{\ensuremath{\operatorname{#1}}}
\newcommand{\tmtextbf}[1]{{\bfseries{#1}}}
\newcommand{\tmtextit}[1]{{\itshape{#1}}}
\newcommand{\tmtexttt}[1]{{\ttfamily{#1}}}
\newenvironment{enumeratealpha}{\begin{enumerate}[a{\textup{)}}] }{\end{enumerate}}
\newenvironment{enumeratenumeric}{\begin{enumerate}[1.] }{\end{enumerate}}
\newenvironment{itemizedot}{\begin{itemize} }{\end{itemize}}
\newenvironment{proof}{\noindent\textbf{Proof\ }}{\hspace*{\fill}$\Box$\medskip}
\definecolor{grey}{rgb}{0.75,0.75,0.75}
\definecolor{orange}{rgb}{1.0,0.5,0.5}
\definecolor{brown}{rgb}{0.5,0.25,0.0}
\definecolor{pink}{rgb}{1.0,0.5,0.5}
\newtheorem{corollary}{Corollary}
\newtheorem{lemma}{Lemma}
\newtheorem{theorem}{Theorem}

\begin{document}

\title{Notes on Bit-reversal Broadcast Scheduling}\author{Marcin
Kik\thanks{\homepage{http://www.im.pwr.wroc.pl/\~{ }kik/}}\\
Faculty of Fundamental Problems of Technology\\
Wroc{\l}aw University of Technology\\
ul. Wybrze\.ze Wyspia\'nskiego 27\\
PL-50-370 Wroc{\l}aw\\
Poland}\maketitle

\begin{abstract}
  This report contains revision and extension of some results about RBO from
  {\cite{DBLP:journals/corr/abs-1108-5095}}.
  RBO is a simple and efficient broadcast scheduling of $n = 2^k$ uniform
  frames for battery powered radio receivers. Each frame contains a key from
  some arbitrary linearly ordered universe. The broadcast cycle -- a sequence
  of frames sorted by the keys and permuted by $k$-bit reversal -- is
  transmitted in a round robin fashion by the broadcaster. \ At arbitrary time
  during the transmission, the receiver may start a simple \ protocol that
  reports to him all the frames with the keys that are contained in a
  specified interval of the key values $[\kappa', \kappa'']$. RBO receives at
  most $2 k + 1$ other frames' keys before receiving \ the first key from
  $[\kappa', \kappa'']$ or noticing that there are no such keys in the
  broadcast cycle. As a simple corollary, \ $4 k + 2$ \ is upper bound the
  number of keys \ outside $[\kappa', \kappa'']$ that will ever be received.
  In unreliable network the expected number of efforts to receive such frames
  is bounded by $(8 k + 4) / p + 2 (1 - p) / p^2$, where $p$ is probability of
  successful reception, and the reception rate of the frames with the keys in
  $[\kappa', \kappa'']$ is $p$ -- the highest possible.
  
  The receiver's protocol state consists of the values $k$, $\kappa'$ and
  $\kappa''$, one wake-up timer and two other $k$-bit variables. Its only
  nontrivial computation -- the computation of the next wake-up time slot --
  can be performed in $O (k)$ simple operations, such as arithmetic/bit-wise
  operations on $k$-bit numbers, using only constant number of \ $k$-bit
  variables.
\end{abstract}

\section{Introduction}

RBO \ {\cite{DBLP:journals/corr/abs-1108-5095}} is a simple and efficient
method of periodic broadcasting of a large sequence of uniform radio messages
for radio receivers with a limited source of energy. Examples of such
receivers are battery powered sensors or portable devices. In modern devices,
the receiver can save the energy by keeping it's radio device switched off for
long periods of time.

The broadcaster transmits in a round robin fashion a large sequence of frames.
Such sequence is called a {\tmem{broadcast cycle}}. Each frame is of the same
length (we call it a {\tmem{time slot}}) and contains in its header a
{\tmem{key}} from an arbitrary linearly ordered universe of key values.

The receiver may decide at arbitrary time (usually somewhere in the middle of
the broadcast cycle) to locate and receive all the frames \ in the stream that
\ contain the keys from some specified range $[\kappa', \kappa'']$. The
receiver may {\tmem{wake-up}} (switch on its radio) at arbitrary time slot to
receive the transmitted frame. However, the radio consumes energy while it is
switched on. We want to minimize the energy dissipated by the receiver, i.e.
to minimize the number of the wake-ups. In RBO, the receiver is able to
receive all the requested frames transmitted since that moment. Roughly
speaking: the receiver listens to some keys of the broadcast cycle and learns
the interval of positions in the sorted sequence with the keys in $[\kappa',
\kappa'']$. After that, it only listens in the time slots that contain the
keys from these positions.

RBO requires that the length of the broadcast cycle is an integer power of
two. This can be achieved by duplicating some of the frames. If $n'$ denotes
the number of frames that must be transmitted, then the length of the
broadcast cycle is $n = 2^k$, where for integer $k$, $k \geq \lceil \log_2 n' \rceil$.

We assume that the length of each frame is the same, i.e. a single time slot.
However, the same key may be repeated many times in the broadcast cycle. Thus,
as single long information attributed with some key can be split among many
frames with the same key. We can also repeat many times, the frames that that
should be delivered more frequently to the receivers. (The frames with the
same key are scattered uniformly over the transmission cycle).

The keys may be arbitrary values from arbitrary linearly ordered domain.
The receiver does not have any knowledge of the distribution of the keys in
the cycle. \ RBO is energetically efficient for the receiver
(Section~\ref{Section-analysis}), robust to the radio interferences
(Section~\ref{Section-unreliable}), and its implementation is very simple and
efficient and requires little memory (Section~\ref{Section-implementation}),
thus it is suitable even for very weak sensor devices (see e.g.
{\cite{TinyOSProgramming}}).

This report updates \ {\cite{DBLP:journals/corr/abs-1108-5095}} as follows:
\begin{itemizedot}
  \item New, simpler  proof of the main theorem (Theorem~\ref{first-hit-Theorem}) is
  based on a simpler decomposition of the time-slots sequence.
  
  \item We focus on the application of the RBO to  filtering the frames with
  the keys from specified interval $[\kappa', \kappa'']$. In
  Corollary~\ref{Corollary-ee} we show that the receiver has to listen to no
  more than $4 k + 2$ frames with keys outside $[\kappa', \kappa'']$, to learn
  which are the time-slots of the frames with keys in $ [\kappa', \kappa'']$.
  
  \item The expected energetic costs for the receiver in unreliable network
  has been estimated in Section~\ref{Section-unreliable}.
  
  \item A simpler and more efficient algorithm for computing the next wake-up
  time slot has been proposed in \ Section~\ref{Section-implementation}.
\end{itemizedot}

\subsection{Example Applications}

The protocol can be applied to the dissemination of information or to
centralized controlling or synchronizing of large populations of energy
constrained devices. Some examples are following:
\begin{itemize}
  \item The keys may be identifiers of records from a huge database
  transmitted in the stream.
  
  \item The keys may be identifiers of the receiver. The broadcaster may send
  commands or messages to individual receivers.
  
  \item The keys may be identifiers of groups of mutually non-interfering
  sensors. Each frame with such key would contain only the header, while the
  rest of the time slot can be used for transmission by the sensors from this
  group.
  
  \item The keys may be coordinates of the objects on the plane encoded by
  Morton z-ordering {\cite{ZorderMorton}}. In such ordering the receiver may
  limit an approximately square region containing the objects that are
  interesting to him.
  
\end{itemize}
 Diverse applications could be mixed within a single stream of frames
  by assigning to them disjoint intervals of key values. The sorted sequence
  of keys is permuted by bit-reversal permutation, \ which scatters the keys
  from each interval uniformly over the whole stream.

\subsection{ Related Work}

Broadcast scheduling for radio receivers with low {\tmem{access time}} (i.e. the
delay to the reception of the required record) and low average {\tmem{tuning
time}} (i.e. the energetic cost) was considered by Imielinski, Viswanathan, and
Badrinath (see e.g. {\cite{DBLP:conf/sigmod/ImielinskiVB94}},
{\cite{DBLP:conf/edbt/ImielinskiVB94}},
{\cite{DBLP:journals/tkde/ImielinskiVB97}}). In
{\cite{DBLP:conf/edbt/ImielinskiVB94}}, {\tmem{hashing}} and {\tmem{flexible
indexing}} for finding single records in broadcast cycle have been proposed
and compared. In {\cite{DBLP:journals/tkde/ImielinskiVB97}}, a distributed
index based on a ordered balanced tree has been proposed. The broadcast
sequence consists of two kinds of {\tmem{buckets}}. Groups of {\tmem{index
buckets}}, containing parts of the index tree, are interleaved with the groups
of {\tmem{data buckets}} containing proper data and a pointer (i.e. time
offset) to the next index bucket. Each group of index buckets consists of the
copy of upper part of the index tree \ together with the relevant fragment of
the lower part of the tree. This mechanism has found useful application even
in more complex scenarios of delivering data to mobile users
{\cite{DBLP:journals/tods/DattaVCK99}}.

Khanna and Zhou {\cite{Khanna2000575}} proposed a sophisticated version of
the index tree aimed at minimizing {\tmem{mean}} access and tuning time, for
given probability of each data record being requested. The broadcast cycle
contains multiple copies of data items, so that spacing between copies of each
item is related to the optimal spacing, minimizing mean access time derived in
{\cite{DBLP:journals/winet/VaidyaH99}}. However the keys are not arbitrary.
The \ key of the item is determined by its probability of being requested. \

Indexing of broadcast stream for \ XML documents
{\cite{DBLP:journals/isci/ChungL07}} or for full text search
{\cite{DBLP:journals/tkde/ChungYK10}} have also been considered.

If the broadcast cycle contains indexing tree structure, then \ the reception
of data in current broadcast cycle depends on the successful reception of the
path to this data. Instead of separate index buckets RBO uses short
{\tmem{headers}} of the frames. Each such header contains the key assigned to
the frame. As a consequence, in unreliable network the receiver has much more
chances of efficient navigation towards the desired frames.

In practical applications, due to imperfect synchronization between the
broadcaster and the receiver, \ the header should also contain either \ the
time-slot number or its bit reversal -- the index of the frame. To enable
changing the contents and the length of the sequence of the transmitted keys
by the broadcaster, the header may also include the parameter $k$, such that
$2^k$ is the length of the broadcast cycle, and some bits used to notify the
receiver that the that the sequence of keys has been changed. For RBO, these
issues have been discussed in \ {\cite{DBLP:journals/corr/abs-1108-5095}}.

Recall that each step of the classic {\tmem{binary search}} algorithm actually
clips the interval of the possible locations of the searched key in the sorted
sequence of keys. The customary presentation is that the keys of the sequence
are organized in a balanced {\tmem{binary search tree}}, and the searched key
is compared with a sequence of keys from subsequent levels of this tree. \
Bit-reversal permutes the sorted sequence of keys so that the broadcast cycle
is a sequence of the subsequent levels of a balanced binary search tree for
the keys. Moreover, each level is recursively so permuted. We show that it
enables efficient search in the periodic transmission of the broadcast cycle
even if the search is started at arbitrary time slot. We also exploit this
property in the computation of the next time slot that should be listened by
the receiver. Bit-reversal permutation has been found useful \ in many
contexts. Some examples of its applications are in FFT algorithm
{\cite{CooleyTukey}} {\cite{CormenLR89}}, lock-free extensible hash arrays
{\cite{DBLP:journals/jacm/ShalevS06}} distributed arrays in P2P
{\cite{DBLP:conf/infocom/FukuchiSSH09}}, address mapping in SDRAM
{\cite{DBLP:conf/scopes/ShaoD05}}, scattering of video bursts in transmission
scheduling in mobile TV {\cite{DBLP:journals/ton/HefeedaH10}}. In RBO,
bit-reversal emerged from updating the recursive definition of the
$\tmop{rbo}$ permutation used in the underlying ranking procedure in
{\cite{DBLP:conf/adhoc-now/Kik08}} in such a way that zero became a fixed
point. The simplicity of bit-reversal computation is a great advantage for
practical implementations.

\section{Notation and preliminaries}

Let $\mathbbm{Z}$ denote the set of integers. Let $\mathbbm{R}$ denote the set
of real numbers. For simplicity and generality, we assume that the keys are
from $\mathbbm{R}$. By \ $[a, b]$ we denote the interval of real numbers $\{x
\in \mathbbm{R}| a \leq x \leq b\}$. If $a > b$ then $[a, b] = \emptyset$. By
$[[a, b]]$ we denote we denote $[a, b] \cap \mathbbm{Z}$ (i.e. interval of
integers between $a$ and $b$). For a set $S$, we denote the number
of its elements by $|S|$.

For $x \in \mathbbm{Z}$, $x \geq 0$, for $i \geq 0$, let $\tmop{bit}_i (x)$ be
the $i$th least significant bit of the binary representation of $x$, i.e.
$\tmop{bit}_i (x) = \lfloor (x \tmop{mod} 2^{i + 1}) / 2^i \rfloor$. For $l
\geq 0$, a number with binary representation $x_l \ldots x_0$ is denoted by
$(x_l, \ldots, x_0)_2$, i.e. $(x_l, \ldots, x_0)_2 = \sum_{i = 0}^l 2^i \cdot
x_i$.

For $x \in [[0, 2^k - 1]]$ let $\tmop{rev}_k (x)$ denote the {\em bit-reversal} of
$x$, i.e: if $x_i = \tmop{bit}_i (x)$ then $x = (x_{k - 1}, x_{k - 2}, \ldots,
x_0)_2$ \ and $\tmop{rev}_k (x) = (x_0, x_1, \ldots, x_{k - 1})_2$.

For a set $S \subseteq [[0, 2^k - 1]]$, $\tmop{rev}_k S$ denotes the
{\tmem{image}} of $S$ under $\tmop{rev}_k$, i.e $\tmop{rev}_k S
=\{\tmop{rev}_k (x) | x \in S\}$.

Let $n$ denote the length of the broadcast cycle, $n = 2^k$, \ for integer $k
\geq 0$. Let $\kappa_{- 1}, \kappa_0, \ldots, \kappa_{n - 1}, \kappa_n$ be a
sequence defined as follows:
\begin{itemizedot}
  \item $\kappa_{- 1} = - \infty$
  
  \item $\kappa_n = + \infty$
  
  \item $\kappa_0, \ldots, \kappa_{n - 1}$ is a sorted sequence of $n$ \
  finite real values of the keys (i.e. $\kappa_i \leq \kappa_{i + 1}$, for \ \
  \ $- 1 \leq i \leq n - 1$).
\end{itemizedot}
Let $\tmop{KEYS} =\{\kappa_0, \ldots, \kappa_{n - 1} \}$ (the set of the
values of the keys in the sequence).

Let $\kappa'$ and $\kappa''$ be finite real key values such that $\kappa' \leq
\kappa''$. $[\kappa', \kappa'']$ is the {\em interval of the searched keys}.

$E [X]$ denotes expected value of random variable $X$.

\subsection{The description of the protocol}

The broadcaster at time-slot $t$ broadcasts the frame with the key
$\kappa_{\tmop{rev}_k (t \tmop{mod} n)}$. The receiver searching for the
$[\kappa', \kappa'']$ has two variables $\tmop{lb}$ and $\tmop{ub}$
initialized to $0$ and \ $n - 1$, respectively. The receiver may start at
arbitrary time slot $s$, and executes the following algorithm:
\begin{itemize}
\item While $\tmop{lb} \leq \tmop{ub}$:
\begin{itemize}
  \item In time-slot $t$ if $\tmop{lb} \leq \tmop{rev}_k (t \tmop{mod} n) \leq
  \tmop{ub}$, then the receiver receives the message with the key $\kappa =
  \kappa_{\tmop{rev}_k (t \tmop{mod} n)}$ and
  \begin{itemize}
    \item if $\kappa < \kappa'$ then it sets $\tmop{lb}$ to $\tmop{rev}_k (t
    \tmop{mod} n) + 1$, else
    
    \item if $\kappa'' < \kappa$ then it sets $\tmop{ub}$ to $\tmop{rev}_k (t
    \tmop{mod} n) - 1$, else
    
    \item if $\kappa' \leq \kappa\leq \kappa''$ then it reports reception of 
     the key $\kappa$ from $[\kappa', \kappa'']$
  \end{itemize}
  \item if $\tmop{lb} > \tmop{ub}$ then the receiver reports that $[\kappa',
  \kappa''] \cap \tmop{KEYS} = \emptyset$
\end{itemize}
\end{itemize}
In the above description we used {\tmem{broadcaster time slot numbers}}. By
{\tmem{receiver time}} we mean the number of time slots that elapsed since
the start of the receiver's protocol. Thus, just before the time slot $s$ the
receiver time is zero, just after time slot $s$ the receiver time is one, an
so on. However, the receiver knows the broadcaster time modulo $n$ (this
information may be included in the frame header) and uses it it to compute the
timer \ waking-up the radio for next reception of the frame.

\subsection{Subsets $Y_{k, s, i}$ and $X_{k, s, i}$}

In the analysis of the receiver's protocol (Section~\ref{Section-analysis}),
we split the sequence of the time slots following the starting slot $s$ into
segments $Y_{k, s, i}$. The set $X_{k, s, i}$ is the set of indexes of the
elements transmitted during time slots $Y_{k, s, i}$. We show that the
``density'' of initially transmitted indexes bounds the length of \
$[\tmop{lb}, \tmop{ub}]$ and the ``sparsity'' of the set of indexes of the
next segment bounds the number of needed receptions. Finally we sum up the
bounds on receptions in all segments. In Section~\ref{Section-implementation},
we use this decomposition and also the binary search tree on the elements of
$X_{k, s, i} $ embedded on the graph of the permutation $\tmop{rev}_k$, for
efficient computation of the wake-up timer.

For the starting time slot $s \in [[0, 2^k - 1]]$, for $i \geq 0$, let $t_{k,
s, i}$ and $l_{k, s, i}$ be defined as follows:
\begin{itemize}
  \item  $t_{k, s, 0} = t$ and \ $l_{k, s, 0} = \max \{l \leq k | t_{k, s, 0}
  \tmop{mod} 2^l = 0\}$.
  
  \item For $i > 0,$ $t_{k, s, i} = (t_{k, s, i - 1} + 2^{l_{k, s, i - 1}})
  \tmop{mod} n$ \ and $l_{k, s, i} = \max \{l \leq k | t_{k, s, i} \tmop{mod}
  2^l = 0\}.$
\end{itemize}
$l_{k, s, i}$ is the maximal length of of the suffix of the zero bits in
binary representation of $t_{k, s, i}$.  
$t_{k, s, i+1}$ is the next time slot after  $t_{k, s, i}$ (modulo $n$),
that has longer such suffix.
Note that $t_{k, s, 0}, t_{k, s, 1}, t_{k, s, 2}, \ldots$ is a (possibly
empty) increasing sequence of some integers from $[[1, 2^k - 1]]$ followed by
infinite sequence of zeroes.

Let $\tmop{last}_{k, s} = \min \{i \geq 0 | t_{k, s, i} = 0\}$. Note that \
$l_{k, s, 0}, \ldots, l_{k, s, \tmop{last}_{k, s}}$ is an increasing sequence
of integers from $[[0, k]]$. For $0 \leq i < \tmop{last}_{k, s}$, let $Y_{k,
s, i} = [[t_{k, s, i}, t_{k, s, i + 1} - 1]]$ and let $Y_{k, s, \tmop{last}} =
[[0, 2^k - 1]]$. For $0 \leq i \leq \tmop{last}_{k, s},$ let $X_{k, s, i} =
\tmop{rev}_k Y_{k, s, i}$.

\begin{lemma}
  \label{X_i-Lemma} \ $X_{k, s, i} =\{\tmop{rev}_k (t_{k, s, i}) + 2^{k -
  l_{k, s, i}} \cdot x' | x' \in [[0, 2^{l_{k, s, i}} - 1]]\}$ and
  $\tmop{rev}_k (t_{k, s, i}) < 2^{k - l_{k, s, i}}$.
\end{lemma}

\begin{proof}
  Let $y_j = \tmop{bit}_j (t_{k, s, i})$, let $l =
  l_{k, s, i}$. Then \ $Y_{k, s, i}$ \ is the set of all numbers $(y_{k - 1},
  \ldots, y_l, y'_{l - 1}, \ldots, y'_0)_2$ such that $y'_j \in \{0, 1\}$.
  Thus $X_{k, s, i} = \tmop{rev}_k (Y_{k, s, i})$ is the set of all numbers
  $(x'_{l - 1}, \ldots, x'_0, y_l, \ldots, y_{k - 1})_2$ such that $x_j' \in
  \{0, 1\}$. Note that $\tmop{rev}_k (t_{k, s, i}) = (0, \ldots, 0, y_l,
  \ldots, y_{k - 1})_2 = (0, \ldots, 0, x_{k - l - 1}, \ldots, x_0)_2$, where
  $x_i = y_{k - 1 - i}$. Thus $\tmop{rev}_k (t_{k, s, i}) < 2^{k - l}$. This
  completes the proof of Lemma~\ref{X_i-Lemma}.
\end{proof}

For $0 \leq i \leq \tmop{last}_{k, s}$, for $0 \leq l \leq l_{k, s, i}$, let
$Y_{k, s, i, l} = [[t_{k, s, i} + \lfloor 2^{l - 1} \rfloor, t_{k, s, i} + 2^l
- 1]]$. Note that $Y_{k, s, i}$ is a disjoint union of the sets $Y_{k, s, i,
l}$, for $0 \leq l \leq l_{k, s, i}$. For $0 \leq i \leq \tmop{last}$, for $0
\leq l \leq l_{k, s, i}$, let $X_{k, s, i, l} = \tmop{rev}_k (Y_{k, s, i,
l})$.

\begin{lemma}
  \label{X_il-Lemma} For $l \in [[0, l_{k, s, i}]]$, \ $X_{k, s, i, l}
  =\{\tmop{rev}_k (t_{k, s, i} + \lfloor 2^{l - 1} \rfloor) + 2^{k - l + 1}
  \cdot x' | x' \in [[0, \lceil 2^{l - 1} \rceil - 1]]\}$ and $\tmop{rev}_k
  (t_{k, s, i} + \lfloor 2^{l - 1} \rfloor) < 2^{k - l + 1}$.
\end{lemma}

\begin{proof}
  If $l = 0$, then $Y_{k, s, i, l} =\{t_{k, s, i}
  \}$ and, \ $X_{k, s, i, l} =\{\tmop{rev}_k (t_{k, s, i})\}$ and
  $\tmop{rev}_k (t_{k, s, i}) < 2^{k + 1}$.
  
  Consider the case: $l > 0$. $Y_{k, s, i, l} = [[(t_{k, s, i} + 2^{l - 1}), (t_{k,
  s, i} + 2^{l - 1}) + 2^{l - 1} - 1]]$. Since $t_{k, s, i} \tmop{mod}
  2^{l_{k, s, i}} = 0$ and $l - 1 < l_{k, s, i}$, we have $(t_{k, s, i} + 2^{l
  - 1}) \tmop{mod} 2^{l - 1} = 0$ and $Y_{k, s, i, l}$ is the set of all
  numbers $(y_{k - 1}, \ldots, y_{l - 1}, y'_{l - 2}, \ldots, y'_0)_2$ such
  that $y_j = \tmop{bit}_j (t_{k, s, i} + 2^{l - 1})$ and $y'_j \in \{0, 1\}$.
  Thus $\tmop{rev}_k (t_{k, s, i} + 2^{l - 1}) < 2^{k - l + 1}$ and $X_{k, s,
  i, l}$ is the set of all numbers $(x'_{l - 2}, \ldots, x'_0, x_{k - l},
  \ldots, x_0)_2$ such that $x_j = y_{k - 1 - j}$ and $x'_j \in \{0, 1\}$. 
\end{proof}

\begin{lemma}
  \label{union-X-Lemma}$\bigcup_{j = 0}^l X_{k, s, i, j} =\{\tmop{rev}_k
  (t_{k, s, i}) + 2^{k - l} \cdot x' | x' \in [[0, 2^l - 1]]\}$ and
  $\tmop{rev}_k (t_{k, s, i}) < 2^{k - l}$.
\end{lemma}

\begin{proof}
  $\bigcup_{j = 0}^l X_{k, s, i, j} =
  \tmop{rev}_k ( \bigcup_{j = 0}^l Y_{k, s, i, j}) = \tmop{rev}_k ([[t_{k, s,
  i}, t_{k, s, i} + 2^l - 1]])$. Since $t_{k, s, i} \tmop{mod} 2^l = 0$, the
  proof follows as in the previous lemmas. 
\end{proof}

\section{\label{Section-analysis}\tmtextbf{The analysis of the receiver's
process for $[\kappa', \kappa'']$}}

Let $r'$ and $r''$ be defined as follows:
\begin{itemizedot}
  \item $r' = \min \{r \in [[- 1, n]] | \kappa' \leq \kappa_r \}$, and
  
  \item $r'' = \max \{r \in [[- 1, n]] | \kappa_r \leq \kappa'' \}$.
\end{itemizedot}
For each $r \in [[r', r'']]$, \ $\kappa_r \in [\kappa', \kappa'']$. If $[\kappa',
\kappa''] \cap \tmop{KEYS} = \emptyset$ then, for some $r \in [[- 1, n - 1]]$,
$\kappa_r < \kappa'$ and $\kappa'' < \kappa_{r + 1}$, and $r' = r + 1$ and
$r'' = r$. \ If $[\kappa', \kappa''] \cap \tmop{KEYS} \not=  \emptyset$ then,
since $- \infty < \kappa' \leq \kappa'' < + \infty$, we have \ $0 \leq r' \leq
r'' \leq n - 1$.

Let $s$ be the first time slot of the receiver's protocol. We assume w.l.o.g.
that $s \in [[0, n - 1]]$. For $t \geq 0$: Let $\tmop{lb}_t$ and $\tmop{ub}_t$
be the values of the variables $\tmop{lb}$ and $\tmop{ub}$, respectively, at
receiver time $t$. (Thus $\tmop{lb}_0 = 0$ and $\tmop{ub}_0 = n - 1$.) Let
$x_t = \tmop{rev}_k ((s + t) \tmop{mod} n)$. Let $\tmop{used}_t = 1$ if
$\tmop{lb}_t \leq x_t \leq \tmop{ub}_t$ and $\tmop{used}_t = 0$ otherwise.
($\tmop{used}_t = 1$ if the receiver wakes-up the radio at receiver time $t$.)
\ Let $\tmop{hit}_t = 1_{}$ if $r' \leq x_t \leq r''$ and $\tmop{hit}_t = 0$
otherwise. ($\tmop{hit}_t = 1$ if the requested frame is received at receiver
time $t$.) The {\tmem{energy}} used in the initial $t$ time slots is
$\tmop{en} (t) = \sum_{j = 1}^t \tmop{used}_j$. The {\tmem{extra energy}} is
the energy used for the reception of messages with the keys outside $[\kappa',
\kappa'']$: $\tmop{ee} (t) = \tmop{en} (t) - \sum_{j = 1}^t \tmop{hit}_j$. Let
\ $\tmop{HY}_t =\{(s + y) \tmop{mod} n | y \in [[0, t - 1]]\}$. $\tmop{HY}_t$
\ is the set of the broadcaster time-slot numbers modulo $n$ of the receiver's
initial $t$ slots. Let $\tmop{HX}_t = \tmop{rev}_k \tmop{HY}_t$. Note that
$\tmop{HX}_0 = \tmop{HY}_0 = \emptyset$. A {\tmem{history of the lower
(respectively, upper) bounds up to time $t$ for $[\kappa', \kappa'']$}} is the
sequence \ $\tmop{HL} (\kappa', \kappa'', t) = (\tmop{lb}_0, \ldots,
\tmop{lb}_t)$ (respectively, $\tmop{HU} (\kappa', \kappa'', t) = (\tmop{ub}_0,
\ldots, \tmop{ub}_t)$).

Let $\tau = \min \{t \geq 0 | \tmop{hit}_t = 1 \vee \tmop{lb}_{t + 1} >
\tmop{ub}_{t + 1} \}$. Note that $\tau$ is the time until the first hit or
noticing that $[\kappa', \kappa''] \cap \tmop{KEYS} = \emptyset$. By the
{\tmem{first receiver cycle}} we mean the first $n$ slots of the receiver
time. \ For $y \in [[0, n]]$, let $\tmop{tt} (y)$ denote the receiver time
just before the transmission of the broadcast time slot $y \tmop{mod} n$, in
the first receiver cycle, i.e. $\tmop{tt} (y) = \min \{t \geq 0 | (s + t)
\tmop{mod} n = y\}$. Note that, since $\tmop{HX}_n = [[0, n - 1]]$, we have
$\tau < n$. \

\begin{theorem}
  \label{first-hit-Theorem}We have $\tau < n$ and $\tmop{en} (\tau) \leq 2
  \cdot k + 1$. \ 
\end{theorem}

\begin{proof}
  We prove Lemmas~\ref{clipped-Lemma}, \ref{first-tree-Lemma},
  \ref{next-tree-Lemma}, \ref{total-Lemma}, to show the theorem in the case
  $[\kappa', \kappa''] \cap \tmop{KEYS} = \emptyset$, and then conclude \ the
  general case.
  
  \begin{lemma}
    \label{clipped-Lemma} If $[\kappa', \kappa''] \cap \tmop{KEYS} =
    \emptyset$, then, for $t \geq 0$, \ $[[\tmop{lb}_t, \tmop{ub}_t]]
    \subseteq [[0, n - 1]] \setminus \tmop{HX}_t $.
  \end{lemma}
  
  \begin{proof}
    Note that the Lemma follows directly from
    the algorithm: $\tmop{lb}_0 = 0$, $\tmop{ub}_0 = n - 1$, and, since
    $[\kappa', \kappa''] \cap \tmop{KEYS} = \emptyset$, for each $x \in
    \tmop{HX}_t$, either $x < \tmop{lb}_t$ or $\tmop{ub}_t < x$. 
  \end{proof}
  
  Since $k$ and $s$ are fixed, we use the following notation: $\tmop{last} =
  \tmop{last}_{k, s}$, $t_i = t_{k, s, i}$, $k_i = l_{k, s, i}$, $Y_i = Y_{k,
  s, i}$, $X_i = X_{k, s, i}$, $Y_{i, j} = Y_{k, s, i, j}$, and $X_{i, j} =
  X_{k, s, i, j}$.
  
  \begin{lemma}
    \label{first-tree-Lemma} If $[\kappa', \kappa''] \cap \tmop{KEYS} =
    \emptyset$, then $\sum_{(s + t) \tmop{mod} n \in Y_0} \tmop{used}_t \leq
    k_0 + 1$.
  \end{lemma}
  
  \begin{proof}
    Since $Y_{0, 0} =\{t_0 \}$, and only the
    first time slot congruent modulo $n$ to $t_0$ is used, we have $\sum_{(s +
    t) \tmop{mod} n \in Y_{0, 0}} \tmop{used}_t = 1$.
    
    For $0 < l \leq k_0$, we show that \ $\sum_{(s + t) \tmop{mod} n \in Y_{0,
    l}} \tmop{used}_t \leq 1$: By Lemma~\ref{clipped-Lemma},
    $[[\tmop{lb}_{2^{l - 1}}, \tmop{ub}_{2^{l - 1}}]] \subseteq [[0, n - 1]]
    \setminus \tmop{HX}_{2^{l - 1}}$, and $\tmop{HX}_{2^{l - 1}} = \bigcup_{j
    = 0}^{l - 1} X_{0, j}$, which, by Lemma~\ref{union-X-Lemma}, contains
    {\tmem{all}} the integers from $[[0, n - 1]]$ \ congruent modulo $2^{k -
    (l - 1)}$ to $\tmop{rev} (t_0)$. Hence $\tmop{ub}_{2^{l - 1}} + 1 \leq
    \tmop{lb}_{2^{l - 1}} - 1 + 2^{k - (l - 1)}$. By Lemma~\ref{X_il-Lemma},
    $X_{0, l}$ contains {\tmem{only}} the integers from $[[0, n - 1]]$ \
    congruent modulo $2^{k - (l - 1)}$ to $\tmop{rev}_k (t_0 + 2^{l - 1})$.
    Hence, $|X_{0, l} \cap [[\tmop{lb}_{2^{l - 1}}, \tmop{ub}_{2^{l - 1}}]] |
    \leq 1$, and, since only the first time slot congruent modulo $n$ is used,
    we have \ $\sum_{(s + t) \tmop{mod} n \in Y_{0, l}} \tmop{used}_t \leq 1$.
    Since $Y_0 = \bigcup_{l = 0}^{k_0} Y_{0, l}$, the Lemma follows.
  \end{proof}
  
  \begin{lemma}
    \label{next-tree-Lemma} If $[\kappa', \kappa''] \cap \tmop{KEYS} =
    \emptyset$, then, for $1 \leq i \leq \tmop{last}$, $\sum_{(s + t)
    \tmop{mod} n \in Y_i} \tmop{used}_t \leq k_i - k_{i - 1} + 1$.
  \end{lemma}
  
  \begin{proof}
    Let $t'' = \tmop{tt} (t_i)$. By
    Lemma~\ref{clipped-Lemma}, $[[\tmop{lb}_{t''}, \tmop{ub}_{t''}]] \subseteq
    [[0, n - 1]] \setminus \tmop{HX}_{t''}$. We have $X_{i - 1} \subseteq
    \tmop{HX}_{t''}$, and, by Lemma~\ref{X_i-Lemma}, $X_{i - 1}$ contains
    {\tmem{all}} the integers from $[[0, n - 1]]$ \ congruent modulo $2^{k -
    k_{i - 1}}$ to $\tmop{rev} (t_{i - 1})$. Thus, $\tmop{ub}_{t''} + 1 \leq
    \tmop{lb}_{t''} - 1 + 2^{k - k_{i - 1}}$. By Lemma~\ref{union-X-Lemma},
    $\bigcup_{j = 0}^{k_{i - 1}} X_{i, j}$ contains {\tmem{only}} the integers
    from $[[0, n - 1]]$ \ congruent modulo $2^{k - k_{i - 1}}$ to
    $\tmop{rev}_k (t_i)$. Hence, we have $| \bigcup_{j = 0}^{k_{i - 1}} X_{i,
    j} \cap [\tmop{lb}_{t''}, \tmop{ub}_{t''}] | \leq 1$, and $\sum_{(t' + t)
    \tmop{mod} n \in \bigcup_{0 \leq j \leq k_{i - 1}} Y_{i, j}} \tmop{used}_t
    \leq 1$.
    
    For $k_{i - 1} < l \leq k_i$, we show that $\sum_{(s + t) \tmop{mod} n
    \in Y_{i, l}} \tmop{used}_t \leq 1$: We have $[[\tmop{lb}_{t'' + 2^{l -
    1}}, \tmop{ub}_{t'' + 2^{l - 1}}]] \subseteq [[0, n - 1]] \setminus
    \tmop{HX}_{t'' + 2^{l - 1}}$ and $\tmop{HX}_{t'' + 2^{l - 1}}$ is a
    super-set of $\bigcup_{j = 0}^{l - 1} X_{i, j}$, which, by
    Lemma~\ref{union-X-Lemma}, contains {\tmem{all}} the integers from $[[0, n
    - 1]]$ \ congruent modulo $2^{k - (l - 1)}$ to $\tmop{rev} (t_i)$. By
    Lemma~\ref{X_il-Lemma}, $X_{i, l}$ contains {\tmem{only}} the integers
    from $[[0, n - 1]]$ \ congruent modulo $2^{k - (l - 1)}$ to $\tmop{rev}_k
    (t_i + \lfloor 2^{l - 1} \rfloor)$. Hence $|X_{i, l} \cap [[\tmop{lb}_{t''
    + 2^{l - 1}}, \tmop{ub}_{t'' + 2^{l - 1}}]] | \leq 1$. 

Thus $\sum_{(s + t)
    \tmop{mod} n \in \bigcup_{k_{i - 1} < j \leq k_i} Y_{i, j}} \tmop{used}_t
    \leq k_i - k_{i - 1}$ and the lemma follows.
  \end{proof}
  
  \begin{lemma}
    \label{total-Lemma} If $[\kappa', \kappa''] \cap \tmop{KEYS} = \emptyset$,
    then $\sum_{t > 0} \tmop{used}_t \leq 2 k + 1$.
  \end{lemma}
  
  \begin{proof}
   $[[0, n - 1]] = \bigcup_{i = 0}^{\tmop{last}}
    Y_i$, and 
$\sum_{t > 0} \tmop{used}_t = $
$\sum_{(s + t) \tmop{mod} n \in
    Y_0} \tmop{used}_t + \sum_{i = 1}^{\tmop{last}} ( \sum_{(s + t) \tmop{mod}
    n \in Y_i} \tmop{used}_t) \text{}$. Thus, by Lemma~\ref{first-tree-Lemma}
    and Lemma~\ref{next-tree-Lemma}, $\sum_{t > 0} \tmop{used}_t \leq k_0 + 1
    + \sum_{i = 1}^{\tmop{last}} (k_i - k_{i - 1} + 1) = k_{\tmop{last}} +
    \tmop{last} + 1$. Since $k_0, \ldots, k_{\tmop{last}}$ is increasing
    sequence of values from $[[0, k]]$, we have $k_{\tmop{last}} \leq k$ and
    $\tmop{last} \leq k$.
  \end{proof}
  
  In Lemma~\ref{total-Lemma} we assumed that $[\kappa', \kappa''] \cap
  \tmop{KEYS} = \emptyset$. Note that we have:
  \begin{itemize}
    \item $0 \leq \tmop{lb}_{\tau} \leq \tmop{ub}_{\tau} \leq n - 1$,
    and
    
    \item $\kappa_{\tmop{lb}_{\tau } - 1} < \kappa' \leq \kappa'' <
    \kappa_{\tmop{ub}_{\tau } + 1}$ \ (since $\tmop{hit}_t = 0$, for $0 < t
    \leq \tau - 1$), \ and
    
    \item $\tmop{HX}_{\tau} \cap [\tmop{lb}_{\tau}, \tmop{ub}_{\tau}] = \emptyset$.
  \end{itemize}
  Since $[\kappa_{\tmop{lb}_{\tau} - 1}, \kappa'] \cap \tmop{KEYS}$ is
  finite, we can choose real number $\gamma$ such that
  $\kappa_{\tmop{lb}_{\tau} - 1} < \gamma < \kappa'$ and $\gamma \not\in
  \tmop{KEYS}$. Since $\kappa_{\tmop{lb}_{\tau} - 1} < \gamma <
  \kappa_{\tmop{ub}_{\tau} + 1}$ , the respective histories of the bounds
  up to the time $\tau$ for $[\kappa', \kappa'']$ and $[\gamma, \gamma]$
  are identical:
  \begin{itemize}
    \item $\tmop{HL} (\kappa', \kappa'', \tau) = \tmop{HL} (\gamma,
    \gamma, \tau)$, and
    
    \item $\tmop{UL} (\kappa', \kappa'', \tau) = \tmop{UL} (\gamma,
    \gamma, \tau)$.
  \end{itemize}
  Note that, since $\tmop{lb}_{\tau} \leq \tmop{ub}_{\tau}$, the
  energy needed to notice that $[\gamma, \gamma]\cap \tmop{KEYS} = \emptyset$
  is at least $\tmop{en} (\tau)$.
  Therefore, by Lemma~\ref{total-Lemma}, $\tmop{en} (\tau) \leq 2 k + 1$. We
  conclude that $\tmop{en} (\tau) \leq 2 k + 1$ also when $[\kappa', \kappa'']
  \cap \tmop{KEYS} \not=  \emptyset$.
\end{proof}

\begin{corollary}
  \label{Corollary-ee}For arbitrary $t > 0$, $\tmop{ee} (t) \leq 4 k + 2$.
\end{corollary}

\begin{proof}
 If $[\kappa', \kappa''] \cap \tmop{KEYS} =
  \emptyset$, then $\tmop{ee} (t) \leq \tmop{en} (\tau)$ and, by
  Theorem~\ref{first-hit-Theorem}, $\tmop{en} (\tau) \leq 2 k + 1$.
  
  Consider the case $[\kappa', \kappa''] \cap \tmop{KEYS} \not=  \emptyset$.
  Then $- 1 < r' \leq r'' < n$. Let $\gamma'$ and $\gamma''$ be such that
  $\kappa_{r' - 1} < \gamma' < \kappa_{r'}$ and $\kappa_{r''} < \gamma'' <
  \kappa_{r'' + 1}$. Then, for arbitrary $t \geq 0$, $\tmop{HL} (\gamma',
  \gamma', t) = \tmop{HL} (\kappa', \kappa'', t)$ and $\tmop{HU} (\gamma'',
  \gamma'', t) = \tmop{HU} (\kappa', \kappa'', t)$. Any reception of the key
  that is outside $[\kappa', \kappa'']$ updates either the lower or the upper
  bound: For $t > 0$, $\tmop{used}_t = 1$ and $\tmop{hit}_t = 0$ if and only
  if either $\tmop{lb}_{t - 1} < \tmop{lb}_t$ or $\tmop{ub}_t < \tmop{ub}_{t -
  1}$. Thus $\tmop{ee} (t)$ is equal to the total number of changes in both
  $\tmop{HL} (\kappa', \kappa'', t)$ and $\tmop{HU} (\kappa', \kappa'', t)$.
  Since $[\gamma', \gamma'] \cap \tmop{KEYS} = \emptyset$, the number of
  changes in $\tmop{HL} (\gamma', \gamma', t)$ is not greater than $2 k + 1$.
  Similarly, the number of changes in $\tmop{HU} (\gamma'', \gamma'', t)$ is
  not greater than $2 k + 1$.
\end{proof}

\subsection{\label{Section-unreliable}Unreliable network}

Consider a model of the network, where the probability of successful reception
is $p$, $0 \leq p \leq 1$. Thus the receiver may wake up to listen in some
time slot, and still fail to receive the frame with probability $q = 1 - p$.
Thus the unit of energy used for the wake-up is lost. We state that in the
case of reception \ failure, the receiver's protocol leaves its variables
$\tmop{lb}$ and $\tmop{ub}$ unchanged and waits for the next time slot from
$\tmop{rev}_k [[\tmop{lb}, \tmop{ub}]]$.

We split the wake-ups of the receiver into {\tmem{hits}} -- the wake-ups in
the time slots from $\tmop{rev}_k [[r', r'']]$, \ and {\tmem{misses}} -- the
remaining wake-ups. The hits are unavoidable: the requested keys are
transmitted during the hits. The penalty for unreliability here is that the
reception rate drops from $1$ to $p$ -- which is the highest possible in this
model. Another penalty is the increase in the number of the misses. We
show the bound on the number of the misses in unreliable network. \ Recall
that the first wake up of the protocol is in time slot $s$. For $t \geq 0$,
let $\tmop{success} (t)$ be true if the transmission in the $t$th receiver's
time slot is successful, and false -- otherwise.

\begin{lemma}
  \label{Lemma-expected-trailer}The expected number of misses after the first
  receiver cycle (i.e. after the initial $n$ time slots) is not greater than
  $2 \cdot q / p^2$.
\end{lemma}

\begin{proof}
  The misses in the cycle following the first cycle are the wake-ups during
  the time slots in $\tmop{rev}_k ([[\tmop{lb}_n, r' - 1]] \cup [[r'' + 1,
  \tmop{ub}_n]])$. The values of $\tmop{lb}_n - 1$ and $\tmop{ub}_n + 1$ are
  the following \ random variables:
  \begin{itemize}
    \item $\tmop{lb}_n - 1 = \max \{- 1\} \cup \{i \in [[0, r' - 1]] \;|\;
    \tmop{success} (\tmop{tt} (\tmop{rev}_k (i))\}$
    
    \item $\tmop{ub}_n + 1 = \min \{n\} \cup \{i \in [[r'' + 1, n - 1]] \;|\;
    \tmop{success} (\tmop{tt} (\tmop{rev}_k (i))\}$
  \end{itemize}
  Each of $r' - (\tmop{lb}_n - 1)$ and $(\tmop{ub}_n + 1) - r''$ can be bound
  by a random variable with geometric distribution (see e.g.
  {\cite{CormenLR89}}) and expected value $1 / p$. Hence, \ $\max \{E [r' -
  \tmop{lb}_n], E [\tmop{ub}_n - r'']\} \leq 1 / p - 1 = 1 / (1 - q) - 1$.
  
  After the $j$th cycle, for $j \geq 1$, each position has been tested $j$
  times. Thus $\max \{E [r' - \tmop{lb}_{j \cdot n}], E [\tmop{ub}_{j \cdot n}
  - r'']\} \leq 1 / (1 - q^j) - 1$ and the expected number of misses in the
  $(j + 1)$st cycle is not greater than $2 (1 / (1 - q^j) - 1)$. Finally, note
  that $\sum_{j = 1}^{\infty} (1 / (1 - q^j) - 1) = \sum_{j = 1}^{\infty} (q^j
  / (1 - q^j)) \leq \frac{1}{1 - q}  \sum_{j = 1}^{\infty} q^j = q / (1 - q)^2
  $.
\end{proof}

The more complex task is to bound the number of misses during \ the first
cycle.

\begin{lemma}
  \label{Lemma-expected-init}If $\tmop{KEYS} \cap [\kappa', \kappa''] =
  \emptyset$, then the expected number of wake-ups (all of them are misses)
  during the first cycle is not greater than $(4 k + 2) / p$.
\end{lemma}

\begin{proof}
  Since $\tmop{KEYS} \cap [\kappa', \kappa''] = \emptyset$, we have $r' = r''
  + 1$. Let us use the notation from the proof of
  Theorem~\ref{first-hit-Theorem}.
  
  First consider the time-slots in \ $Y_0$. There is one wake-up in $Y_{0, 0}
  =\{t_0 \}$. For each \ $l \geq 0$, \ $\bigcup_{j = 0}^l X_{0, j} \subseteq
  \tmop{HX}_{\tmop{tt} (\min Y_{0, l + 1})}$. \ Hence, by
  Lemma~\ref{union-X-Lemma}, 
\begin{itemize}
  \item
  $\tmop{lb}_{\tmop{tt} (\min Y_{0, l + 1})} - 1
  \geq 
  \max \{- 1\} \cup \{i \in [[0, r' - 1]] \;|\; (i - \tmop{rev}_k (t_0))
  \tmop{mod} 2^{k - l} = 0 \wedge \tmop{success} (\tmop{tt} (\tmop{rev}_k
  (i)))\}$, and
\item 
  $\tmop{ub}_{\tmop{tt} (\min Y_{0, l + 1})} + 1 
  \leq 
  \min \{n\}
  \cup \{i \in [[r'' + 1, n - 1]] \;|\; (i - \tmop{rev}_k (t_0)) \tmop{mod} 2^{k -
  l} = 0 \wedge \tmop{success} (\tmop{tt} (\tmop{rev}_k (i)))\}$.
\end{itemize}
 Note that
  $[[0, r' - 1]] \cup [[r'' + 1, n - 1]] = [[0, n - 1]]$. Thus, $E
  [(\tmop{ub}_{\tmop{tt} (\min Y_{0, l + 1})} - \tmop{lb}_{\tmop{tt} (\min
  Y_{0, l + 1})}) / 2^{k - l}] < 2 / p$ -- the expected number of integers
  congruent modulo $2^{k - l}$ to $\tmop{rev}_k (t_0)$ in
  [[$\tmop{lb}_{\tmop{tt} (\min Y_{0, l + 1})}, \tmop{ub}_{\tmop{tt} (\min
  Y_{0, l + 1})}]]$. Since, by Lemma~\ref{X_il-Lemma}, all elements of $X_{0,
  l + 1}$ are congruent modulo $2^{k - l}$ to $\tmop{rev}_k (t_0 + \lfloor 2^l
  \rfloor)$, the expected number of wake-ups during time slots $Y_{0, l + 1}$
  is bounded by $2 / p$. Thus the expected number of wake-ups in $Y_0$ is not
  greater than $2 k_0 / p + 1 \leq 2 (k_0 + 1) / p$.
  
  Now consider $Y_i$, for $i \in [[1, \tmop{last}]]$. Since $X_{i - 1}
  \subseteq \tmop{HX}_{\tmop{tt} (\min Y_i)}$ and, by Lemma~\ref{X_i-Lemma},
  $X_{i - 1}$ contains all integers congruent modulo $2^{k - k_{i - 1}}$ to
  $\min X_{i - 1}$ and, by Lemma~\ref{union-X-Lemma}, $\bigcup_{j = 0}^{k_{i -
  1}} X_{i, j}$, contains only integers congruent modulo $2^{k - k_{i - 1}}$
  to $\min X_i$, the expected number of wake-ups in $ \bigcup_{j = 0}^{k_{i -
  1}} Y_{i, j}$ can be bound, as above, by $2 / p$.
  
  For each $l \in [[k_{i - 1} + 1, k_i]]$, we use $\bigcup_{j = 0}^{l - 1}
  X_{i, j} \subseteq \tmop{HX}_{\tmop{tt} (\min Y_{i, l})}$, to bound the
  expected number of wake-ups in $Y_{i, l}$ by $2 / p$. Thus the expected
  number of wake-ups in $Y_i$ is not greater than $2 (k_i - k_{i - 1} + 1) /
  p$.
  
  Summing up, as in the proof of Lemma~\ref{total-Lemma}, the expected
  number of wake-ups during the first cycle is at most $\frac{2}{p} (k_0 + 1 +
  \sum_{i = 0}^{\tmop{last}} (k_i - k_{i - 1} + 1)) \leq (4 k + 2) / p$.
\end{proof}

\begin{theorem}
  \label{Theorem-expected}The expected number of misses during the infinite
  execution of the protocol is not greater than $(8 k + 4) / p + 2 (1 - p) /
  p^2$.
\end{theorem}

\begin{proof}
  If $\tmop{KEYS} \cap [\kappa', \kappa''] = \emptyset$, then the theorem
  follows directly from Lemmas~\ref{Lemma-expected-trailer} and
  \ref{Lemma-expected-init}.
  
  Consider the case $\tmop{KEYS} \cap [\kappa', \kappa''] \not=  \emptyset$.
  As in \ Corollary~\ref{Corollary-ee}, let $\gamma'$ and $\gamma''$ be key
  values such that $\kappa_{r' - 1} < \gamma' < \kappa_{r'}$ and $\kappa_{r''}
  < \gamma'' < \kappa_{r'' + 1}$. Let $\tmop{EX}_{\gamma}$ denotes the
  expected number of misses in the first cycle when the protocol is started
  for interval $[\gamma, \gamma]$. By Lemma~\ref{Lemma-expected-trailer},
  $\max \{E_{\gamma'}, E_{\gamma''} \} \leq (4 k + 2) / p$. The expected
  number of misses during the first cycle of the protocol for $[\kappa',
  \kappa'']$ is the sum of the expected number of misses on both sides of
  $[r', r'']$ which is not greater than $E_{\gamma'} + E_{\gamma''}$. 
\end{proof}

\section{\label{Section-implementation}Implementation issues}

We present an efficient algorithm for computing the time slot of the reception
of the next frame required by the protocol. The efficiency of this algorithm
is based on the observation that elements of $X_{k, s, i}$ are organized by
$\tmop{rev}_k$ into subsequent levels of an almost balanced binary search
tree.

\subsection{Binary search tree on $X_{k, s, i}$}

For $d \geq 0$, for any sequence $c = (c_1, \ldots, c_d) \in \{- 1, 1\}^d$,
let a {\tmem{descendant}} of $x$ by path $c$ be defined as 
$\tmop{dsc}_k(x, c) = x + \sum_{i = 1}^d 2^{k - i} \cdot c_i$. 
Note that $\tmop{dsc}_k (x,(c_1, c_2, \ldots, c_d)) = \tmop{dsc}_{k - 1} (\tmop{dsc}_k (x, (c_1)), (c_2,
\ldots, c_d))$. Note that $(\tmop{dsc}_k (x, (c_1, \ldots, c_d)) - x)
\tmop{mod} 2^{k - d} = 0$. Let a {\tmem{level at depth $d$ rooted at $x$}} be
defined as $L_{k, d} (x) =\{\tmop{dsc}_x (x, c) | c \in \{- 1, 1\}^d \}$. Let
a {\tmem{sub-tree of depth $d$ rooted at $x$}} be defined as $\tmop{ST}_{k, d}
(x) =\{\tmop{dsc}_x (x, c) | \exists_{d' \in [[0, d]]} c \in \{- 1, 1\}^{d'}
\}$. The following properties are easy to note without the proof:

\begin{lemma}
  \label{BST-properties}For $k \geq 0$, \ for \ $d \in [[0, k]]$, we have the
  following properties:
  \begin{enumeratealpha}
    \item $|L_{k, d} (x) | = 2^d$.
    
    \item \label{BST-L-ST}$L_{k, 0} (x) = \tmop{ST}_{k, 0} (x) =\{x\}$ and,
    for $d > 1$, $L_{k, d} (x) = \tmop{ST}_{k, d} (x) \setminus \tmop{ST}_{k,
    d - 1} (x)$.
    
    \item $| \tmop{ST}_{k, d} (x) | = 2^{d + 1} - 1$.
    
    \item $\tmop{ST}_{k, d} (x) =\{x + i \cdot 2^{k - d} | i \in [[- 2^d + 1,
    2^d - 1]]\}$.
    
    \item \label{BST-with-right-ST}If $d \geq 1$ then $\{x\} \cup \tmop{ST}_{k
    - 1, d - 1} (\tmop{dsc}_k (x, (1))) =\{x + i \cdot 2^{k - d} | i \in [[0,
    2^d - 1]]\}$.
    
    \item \label{BST-ST-root-ST} $\tmop{ST}_{k, d} (x) = \tmop{ST}_{k - 1, d -
    1} (\tmop{dsc}_k (x, (- 1))) \cup \{x\} \cup \tmop{ST}_{k - 1, d - 1} (x,
    \tmop{dsc}_k (x, (1)))$.
    
    \item \label{BST-inorder} $\max \tmop{ST}_{k - 1, d} (\tmop{dsc}_k (x, (-
    1))) + 2^{k - 1 - d} = x = \min \tmop{ST}_{k - 1, d} (\tmop{dsc}_k (x,
    (1))) - 2^{k - 1 - d}$.
  \end{enumeratealpha}
\end{lemma}

Lemma~\ref{X-BST} shows that each $X_{k, s, i}$ is organized by $\tmop{rev}_k$
in a binary search tree with the root at $\min X_{k, s, i} = \tmop{rev}_k
(t_{k, s, i})$, without the left sub-tree and with a totally balanced right
sub-tree, see Figure~\ref{BST-fig}.

Lemma~\ref{X-BST}\ref{parent-above-child} states that the elements of the 
levels closer to the root
have lower values of their $k$-bit reversals than 
the elements of the more distant levels.

\begin{figure}[h]
\centering{\resizebox{15cm}{!}{
  \includegraphics{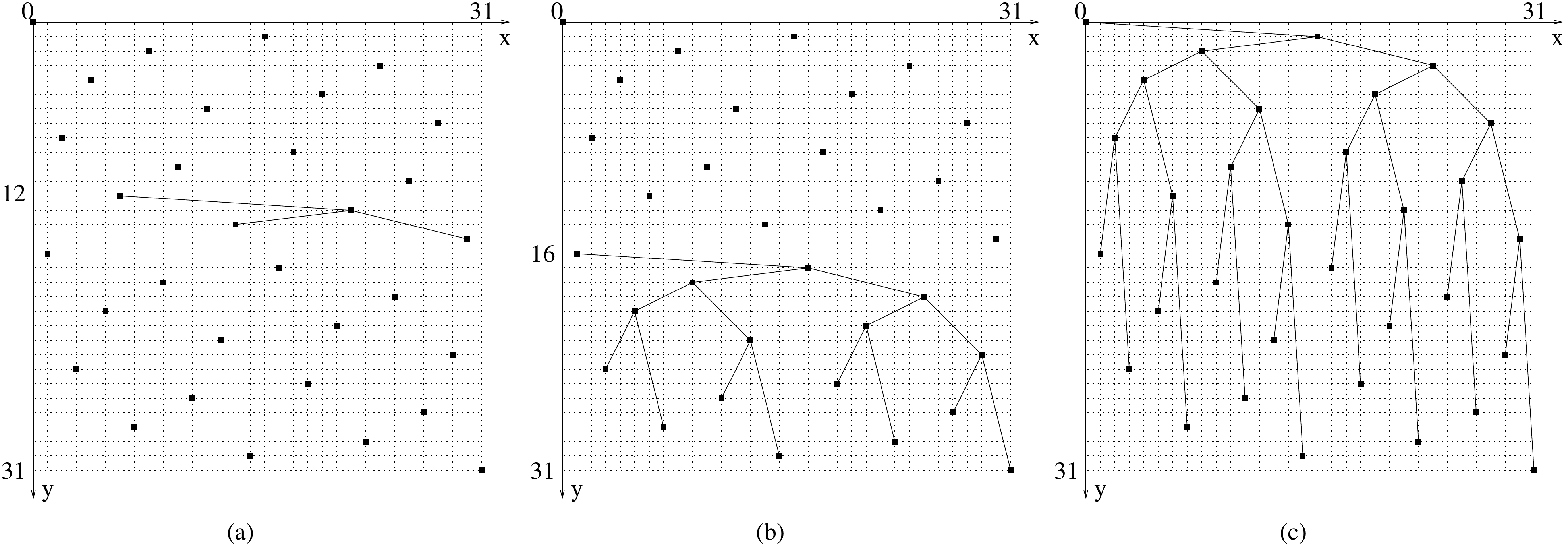}
}}
  \caption{\label{BST-fig}The binary search trees for $X_{5, 12, 0}$ (a),
  $X_{5, 12, 1}$ (b) and $X_{5, 12, 2}$ (c), on the graph of $y = \tmop{rev}_5
  (x)$. Note that the $y$ axis of the graph is directed downwards.}
\end{figure}

\begin{lemma}
  \label{X-BST} For $k \geq 0$, $t \in [[0, 2^k - 1]]$, $i \in [[0,
  \tmop{last}_{k, s}]]$, let $r = \tmop{rev}_k (t_{k, s, i})$ and $l = k_{k,
  s, i}$.
  
  Then we have:
  \begin{enumeratealpha}
    {\color{black} \item \label{upper-levels} $X_{k, s, i, 0} =\{r\}$ and, for
    $\text{$d \in [[1, l]]$}, \bigcup_{j = 0}^d X_{k, s, i, j} =\{r\} \cup
    \tmop{ST}_{k - 1, d - 1} (\tmop{dsc}_k (r, (1)))$.}
    
    \item \label{X-BST-structure} If $l > 0$ then $X_{k, s, i} =\{r\} \cup
    \tmop{ST}_{k - 1, l - 1} (\tmop{dsc}_k (r, (1)))$. If $l = 0$ then $X_{k,
    s, i} =\{r\}$.
    
    {\color{black} \item \label{X-BST-level}$X_{k, s, i, 0} =\{r\}$ and, for
    $d \in [[1, l]]$, \ $X_{k, s, i, d} = L_{k - 1, d - 1} (\tmop{dsc}_k (r,
    (1)))$.}
    
    \item \label{parent-above-child}If $c' \in \{- 1, 1\}^{d'}$ and $c'' \in
    \{- 1, 1\}^{d''}$, where $0 \leq d' < d'' \leq l$, and $x' = \tmop{dsc}_k
    (r, c')$ and $x'' = \tmop{dsc}_k (r, c'')$ and \ $x', x'' \in X_{k, s, i}$
    and $y' = \tmop{rev}_k (x')$ and $y'' = \tmop{rev}_k (x'')$ then $y' <
    y''$.
  \end{enumeratealpha}
\end{lemma}

\begin{proof}

  Lemma~\ref{X-BST}\ref{upper-levels}: By
  Lemma~\ref{BST-properties}\ref{BST-with-right-ST} $\{r\} \cup \tmop{ST}_{k
  - 1, d - 1} (\tmop{dsc}_k (r, (1))) =\{r + i \cdot 2^{k - d} | i \in [[0,
  2^d]]\}$ which, by Lemma~\ref{union-X-Lemma}, is equal to $\bigcup_{j = 0}^d
  X_{k, s, i, j}$.
  
  Lemma~\ref{X-BST}\ref{X-BST-structure} follows from $X_{k, s, i} =
  \bigcup_{j = 0}^l X_{k, s, i, j}$ and from
  Lemma~\ref{X-BST}\ref{upper-levels}.
  
  Lemma~\ref{X-BST}\ref{X-BST-level} follows from \
  Lemma~\ref{X-BST}\ref{upper-levels} and from
  Lemma~\ref{BST-properties}\ref{BST-L-ST}.
  
  Lemma~\ref{X-BST}\ref{parent-above-child}: If $d' = 0$ then $x' = r$ and
  the lemma follows, since $t_{k, s, i} = \min Y_{k, s, i}$. Otherwise, we
  have \ $0 < d' < d''$, $x' \not=  r$ and $x'' \not=  r$ and, by
  Lemma~\ref{X-BST}\ref{X-BST-structure}, $x', x'' \in \tmop{ST}_{k - 1, l -
  1} (\tmop{dsc}_k (r, (1)))$. Thus $x' \in L_{k - 1, d' - 1} (\tmop{dsc}_k
  (r, (1)))$ and $x'' \in L_{k - 1, d'' - 1} (\tmop{dsc}_k (r, (1)))$. By
  Lemma~\ref{X-BST}\ref{X-BST-level}, $x' \in X_{k, s, i, d'}$ \ and $x'' \in
  X_{k, s, i, d''}$. To conclude, note that $\max Y_{k, s, i, d'} < \min Y_{k,
  s, i, d''}$. 
\end{proof}

\subsection{\label{Section-nsi}Implementation of $\tmop{nsi}$}

In realistic implementation, after each reception, the receiver has to compute
the next time slot with the index of the transmitted key \ in the interval
$[\tmop{lb}, \tmop{ub}]$, and switch off the radio for the time remaining to
this event.

By $\tmop{nsi}_k (t, r_1, r_2)$ we denote the next slot number (modulo $2^k$)
after the slot $t$ with its $k$-bit reversal in $[r_1, r_2]$: \ For \ $r_1,
r_2 \in [[0, 2^k - 1]]$, $r_1 \leq r_2$, and $t \in [[0, 2^k - 1]]$,
$\tmop{nsi}_k (t, r_1, r_2) = (t + \tau (t, r_1, r_2)) \tmop{mod} 2^k$, \
where $\tau (t, r_1, r_2) = \min \{d > 0 | \tmop{rev}_k ((t + d) \tmop{mod}
2^k) \in [[r_1, r_2]]\}$.

One could naively test subsequent values after $t$ or all values in
$\tmop{rev}_k [[r_1, r_2]]$. However, both these methods are time consuming,
when both $2^k / (r_2 - r_1)$ and $r_2 - r_1$ are large.

We present an efficient algorithm for the computation of $\tmop{nsi}_k (t,
r_1, r_2)$:
\begin{enumeratenumeric}
  \item $t'' \leftarrow (t + 1) \tmop{mod} 2^k$
  
  \item $l \leftarrow 0$
  
  \item \label{repeat-line}repeat
  \begin{enumeratealpha}
    \item $t' \leftarrow t''$
    
    \item \label{internal-while}while $l < k \wedge t' \tmop{mod} 2^{l + 1} =
    0$ \ do \ $l \leftarrow l + 1$
    
    \item $x_1 \leftarrow \tmop{rev}_k (t')$
    
    \item $t'' \leftarrow (t' + 2^l) \tmop{mod} 2^k$
    
    \item $x_2 \leftarrow \tmop{rev}_k (t' + 2^l - 1)$
  \end{enumeratealpha}
  \item \label{until-line}until $r_1 \leq x_2 \wedge r_2 \geq x_1 \wedge
  \lceil (r_1 - x_1) / 2^{k - l} \rceil \leq \lfloor (r_2 - x_1) / 2^{k - l}
  \rfloor$
  
  \item \label{bin-search-start}$c \leftarrow 2^{k - 1}$
  
  \item \label{bin-search-while}while $x_1 < r_1 \vee x_1 > r_2$ do
  \begin{enumeratealpha}
    \item if $x_1 < r_1$ then $x_1 \leftarrow x_1 + c$ else $x_1 \leftarrow
    x_1 - c$
    
    \item $c \leftarrow c / 2$
  \end{enumeratealpha}
  \item return $\tmop{rev}_k (x_1)$
\end{enumeratenumeric}

Correctness of the algorithm:

Let $s = (t + 1) \tmop{mod} 2^k$. Let the iterations of the ``repeat-until''
loop be numbered starting from zero. After the $i$th iteration, at line
\ref{until-line}, we have $l = l_{k, s, i}$, $t' = t_{k, s, i}$, \ $x_1 = \min
X_{k, s, i} = \tmop{rev}_k (t_{k, s, i})$, $x_2 = \max X_{k, s, i}$, and $t''
= t_{k, s, i + 1}$. Let $i' = \min \{i \geq 0 | X_{k, s, i} \cap [r_1, r_2]
\not=  \emptyset\}$. Since $r_1, r_2 \in [[0, 2^k - 1]]$, \ $r_1 \leq r_2$,
and $X_{k, s, \tmop{last}_{k, s}} = [[0, 2^k - 1]]$, we have $0 \leq i' \leq
\tmop{last}_{k, s}$. Thus, by Lemma~\ref{X_i-Lemma}, $i'$ is the number of the
first iteration, after which$ r_1 \leq x_2 \wedge r_2 \geq x_1 \wedge \min \{j
| x_1 + 2^{k - l} \cdot j \geq r_1 \} \leq \max \{j | x_1 + 2^{k - l} \cdot j
\leq r_2 \}$, which is equivalent to $ r_1 \leq x_2 \wedge r_2 \geq x_1 \wedge
\lceil (r_1 - x_1) / 2^{k - l} \rceil \leq \lfloor (r_2 - x_1) / 2^{k - l}
\rfloor$.

After the ``repeat-until'' loop finishes, at line \ref{bin-search-start}, we
have $x_1 = \tmop{rev}_k (t_{k, s, i'})$ and, by
Lemma~\ref{X-BST}\ref{X-BST-structure}, $X_{k, s, i'} =\{x_1 \} \cup S$,
where either $S = \tmop{ST}_{k - 1, l - 1} (\tmop{dsc}_k (x_1, (1)))$, if $l >
0$, or $S = \emptyset$, if $l = 0$. Since $X_{k, s, i'} \cap [r_1, r_2]
\not= \emptyset$, we do a binary search in $X_{k, s, i'}$ until we enter the
interval $[r_1, r_2]$ for the first time. By the
Lemma~\ref{X-BST}\ref{parent-above-child}, the returned value is $\min
\{\tmop{rev}_k (x) | x \in X_{k, s, i'} \}$.

Complexity of the algorithm:

The memory complexity: Only the constant number of $k$-bit variables are used.

The time complexity: The number of iterations of the ``repeat-until'' loop is
never greater than $k + 1$. Since the value of $l$ never decreases, the total
number of iterations of the internal ``while'' loop
(line~\ref{repeat-line}\ref{internal-while}) in all iterations of the
``repeat-until'' loop is never grater than $k + 1$. The total number of
iterations of the binary search loop (starting at line~\ref{bin-search-while})
is never greater than $k$. Thus the total complexity is $O (k)$ elementary
operations on $k$-bit integers.

Multiplication, division and modulo operations by the powers of two can be
replaced by shifting or bit-masking operations. The implementation of this
algorithm in programming language, with optimizations of bit-wise operations
can be found on {\cite{RBO-WWW}}.

Some technical aspects of the implementation, such as dealing with imperfect
synchronization and proposed structure of the frame header has been discussed
in technical report {\cite{DBLP:journals/corr/abs-1108-5095}}.

\end{document}